\documentclass[a4paper]{article}

\usepackage{authblk}
\usepackage{amsmath}
\usepackage{amssymb}
\usepackage{amsthm}
\usepackage[pdftex]{hyperref}

\newtheorem{defn}{Definition}
\newtheorem{thm}{Theorem}
\newtheorem*{thm*}{Theorem}

\makeatletter
\newcommand{\xdashrightarrow}[2][]{\ext@arrow 0359\rightarrowfill@@{#1}{#2}}
\def\rightarrowfill@@{\arrowfill@@\relax\relbar\rightarrow}
\def\arrowfill@@#1#2#3#4{%
  $\m@th\thickmuskip0mu\medmuskip\thickmuskip\thinmuskip\thickmuskip
   \relax#4#1
   \xleaders\hbox{$#4#2$}\hfill
   #3$%
}
\makeatother

\DeclareMathOperator{\degto}{\xdashrightarrow{SLOCC}}
\DeclareMathOperator{\slto}{\xrightarrow{SLOCC}}
\DeclareMathOperator{\mean}{\mathbb{E}}

\DeclareMathOperator{\prob}{Pr}
\DeclareMathOperator{\GHZ}{GHZ}
\DeclareMathOperator{\EPR}{EPR}
\DeclareMathOperator{\mamu}{MaMu}

\newcommand{\ket}[1]{\left|#1\right\rangle}

\title{Entanglement distillation from Greenberger-Horne-Zeililinger shares}
\author[1]{P\'eter Vrana}
\affil[1]{Department of Geometry, Budapest University of Technology and Economics, Egry J\'ozsef u. 1., 1111 Budapest, Hungary}
\author[2]{Matthias Christandl}
\affil[2]{Department of Mathematical Sciences, University of Copenhagen, Universitetsparken 5, 2100 Copenhagen, Denmark}

\begin{document}
\maketitle

\begin{abstract}
We study the problem of converting a product of Greenberger-Horne-Zeilinger (GHZ) states shared by subsets of several parties in an arbitrary way into GHZ states shared by every party. Our result is that if SLOCC transformations are allowed, then the best asymptotic rate is the minimum of bipartite log-ranks of the initial state. This generalizes a result by Strassen on the asymptotic subrank of the matrix multiplication tensor.
\end{abstract}

\section{Introduction}\label{sec:intro}

It has been realized recently (see e.g. \cite{tripartite,VranaChristandl}) that certain problems in algebraic complexity theory can be interpreted as questions regarding asymptotic entanglement transformations by stochastic local operations and classical communication (SLOCC). Most prominently, the asymptotic rate at which GHZ states can be converted to triples of Einstein-Podolsky-Rosen (EPR) pairs, shared between parties AB, BC and AC, is the same as the exponent of matrix multiplication $\omega$, i.e. the infimum of real numbers $\tau$ such that $n\times n$ matrices can be multiplied using $O(n^\tau)$ arithmetic operations \cite{tripartite}. This particular problem is still open, the best bounds currently being $2\le\omega\le 2.3728639$ \cite{LeGall}, but there are several similar problems where the exact value is known. Interestingly, these include the reverse transformation. To state this precisely, we first introduce some notation.
\begin{defn}
Let $V_1,\ldots,V_k$ and $W_1,\ldots,W_k$ be vector spaces and $\psi\in V_1\otimes\cdots\otimes V_k$, $\phi\in W_1\otimes\cdots\otimes W_k$. We say that $\psi$ can be transformed into $\phi$ via SLOCC ($\psi\slto\phi$) if there exist linear transformations $A_i:V_i\to W_i$ such that $(A_1\otimes\cdots\otimes A_k)\psi=\phi$.

The asymptotic SLOCC conversion rate from $\psi$ to $\phi$ is
\begin{equation}
\omega(\psi,\phi)=\lim_{n\to\infty}\frac{1}{n}\inf\left\{m\in\mathbb{N}|\psi^{\otimes m}\slto\phi^{\otimes n}\right\}.
\end{equation}
\end{defn}
It is clear that the definition is not sensitive to multiplication by scalars, so for simplicity, we will work with unnormalized states.

Let $\mamu$ denote the triple of EPR pairs, i.e.
\begin{equation}
\mamu=\EPR_{AB}\otimes\EPR_{BC}\otimes\EPR_{AC}
\end{equation}
and let $\GHZ=\ket{000}+\ket{111}$. Then \cite[Theorem 6.6]{Strassen1}
\begin{equation}
\omega(\mamu,\GHZ)=\frac{1}{2}
\end{equation}

We review the main ingredients of the proof. The first one is a relaxation of SLOCC convertibility:
\begin{defn}
Given tensors $\psi\in V_1\otimes\cdots\otimes V_k$, $\phi\in W_1\otimes\cdots\otimes W_k$, we say that $\psi$ degenerates to $\phi$ if there exist linear transformations $A_i(\epsilon):V_i\to W_i$ depending on $\epsilon$ and $\epsilon^{-1}$ polynomially (i.e. the matrix element functions are polynomials) such that $(A_1(\epsilon)\otimes\cdots\otimes A_k(\epsilon))\psi=\phi+O(\epsilon)$ as $\epsilon\to 0$. This fact will be denoted by $\psi\degto\phi$.
\end{defn}
The key result relating degeneration to asymptotic SLOCC transformations is that $\psi\degto\phi$ implies $\omega(\psi,\phi)\le 1$ (see \cite{Bini,Strassen2} or \cite{VranaChristandl} for a proof using the same notations as here). As an example, $\GHZ\degto W$ leads to a simple proof of $\omega(\GHZ,W)=1$, even though $\GHZ$ cannot be converted to $W$ via SLOCC in a one-shot setting.

Next, $\mamu^{\otimes n}$ is the same as
\begin{equation}\label{eq:mamu}
\sum_{i_1,i_2,i_3=1}^{N}\ket{i_1i_2}\ket{i_2i_3}\ket{i_3i_1}
\end{equation}
with $N=2^n$ up to relabelling the local basis states. For some $g\in\mathbb{N}$, we choose the local linear transformations in the following way \cite{Strassen1}:
\begin{equation}
\begin{split}
A_1(\epsilon)\ket{i_1i_2} & = \epsilon^{i_1^2+2i_1i_2}\ket{i_1i_2}  \\
A_2(\epsilon)\ket{i_2i_3} & = \epsilon^{i_2^2+2i_2(i_3-g)}\ket{i_2i_3}  \\
A_3(\epsilon)\ket{i_3i_1} & = \epsilon^{(i_3-g)^2+2(i_3-g)i_1}\ket{i_3i_1}.
\end{split}
\end{equation}
Applying the product to a term $\ket{i_1i_2}\ket{i_2i_3}\ket{i_3i_1}$ multiplies it by $\epsilon^{(i_1+i_2+i_3-g)^2}$, therefore we have
\begin{equation}
\sum_{i_1,i_2,i_3=1}^{N}\ket{i_1i_2}\ket{i_2i_3}\ket{i_3i_1}\degto\sum_{\substack{1\le i_1,i_2,i_3\le N  \\  i_1+i_2+i_3=g}}\ket{i_1i_2}\ket{i_2i_3}\ket{i_3i_1}.
\end{equation}
Up to relabelling, the resulting state is a GHZ state. To see this, it is enough to note that every local basis element appears at most once, since any two of $i_1$, $i_2$ and $i_3$ determines the third one uniquely.

To complete the proof, one only needs to choose $g$ as a function of $N$ in such a way that the number of terms grows as quadratically. This can be ensured by choosing $g=N$.

One possible generalization of the matrix multiplication state to $k>3$ parties can be obtained by replacing the EPR pairs by arbitrary collections of GHZ states shared by subsets of the parties. The pattern can be conveniently encoded in a hypergraph on $\{1,\ldots,k\}$ with multiple edges allowed. For example, the matrix multiplication state corresponds to the complete graph $K_3$. The statement of our main result involves the concept of edge-connectivity. For the readers' convenience, we recall the definition here.
\begin{defn}
A hypergraph $H$ is connected if for any pair of vertices $x$ and $y$ there is a sequence $(x=v_0,e_1,v_1,e_2,v_2,\ldots,e_n,v_n=y)$ such that the $e_i$ are hyperedges, $v_i$ are vertices and $e_i$ is incident to $v_{i-1}$ and $v_i$ for $i=1,\ldots,n$.

A hypergraph is $l$-edge-connected if it remains connected after removing any subset of strictly less than $l$ edges. The edge-connectivity $\lambda(H)$ of a hypergraph $H$ is the largest $l$ such that $H$ is $l$-edge-connected.
\end{defn}

Now we are ready to state our main result.
\begin{thm}\label{thm:mainres}
Let the state corresponding to the hypergraph $H$ be $\GHZ^H$. Then
\begin{equation}\label{eq:GHZrate}
\omega(\GHZ^H,\GHZ)=\frac{1}{\lambda(H)}
\end{equation}
\end{thm}

Note that the right hand side is easily seen to be a lower bound. This follows from the fact that 1) the bipartite rank across any bipartition cannot be increased asymptotically, 2) the bipartite log-rank of $\GHZ^H$ across the bipartition $S$-$\overline{S}$ (where $\overline{S}$ denotes the complement) is the number of hyperedges having a nonvanishing intersection with both $S$ and $\overline{S}$, 3) the minimum of these ranks is therefore $\lambda(H)$, and 4) the rank of $\GHZ$ over any bipartition is $2$.

\section{Proof of main result}\label{sec:proof}

By a hypergraph we mean a triple $(V,E,I)$ where $V$ and $E$ are sets and $I\subseteq V\times E$. Elements of $V$ are called vertices, $E$ is the set of edges, and $v\in V$ is said to be incident with $e\in E$ if $(v,e)\in I$. The sum of two hypergraphs $H_1=(V,E_1,I_1)$ and $H_2=(V,E_2,I_2)$ on a common vertex set $V$ is defined as $H_1+H_2=(V,E_1\sqcup E_2,I_1\sqcup I_2)$, where $\sqcup$ stands for disjoint union and $I_1\sqcup I_2$ is thought of as a subset of $V\times(E_1\sqcup E_2)$ in the obvious way. Any hypergraph can be uniquely written as the sum of hypergraphs having one edge each.

If $V=[k]=\{1,2,\ldots,k\}$ for some positive integer $k$, we define the states $\GHZ^H_r$ by requiring $\GHZ^{H_1+H_2}_r=\GHZ^{H_1}_r\otimes\GHZ^{H_2}_r$ and that for the hypergraph $H_S$ which consists of a single edge $S\subseteq[k]$, the state $\GHZ^{H_S}_r$ is an $r$-level GHZ state shared among the parties in $S$, i.e.
\begin{equation}
\GHZ^{H_S}_r=\ket{00\ldots 0}_{\overline{S}}\otimes(\ket{11\ldots 1}_S+\ket{22\ldots 2}_S+\cdots+\ket{rr\ldots r}_S)
\end{equation}
These states are well-defined only up to SLOCC equivalence, but this is sufficient for our purposes. We will make use of the fact that
\begin{equation}
\omega(\GHZ^H_a,\GHZ^H_b)=\frac{\log b}{\log a}
\end{equation}
for any hypergraph $H$ having at least one edge of size at least $2$. If $|E|=1$ and the only edge is incident with every vertex, we will also write $\GHZ_r$ instead of $\GHZ^H_r$.

Similarly to eq. \eqref{eq:mamu}, the states $\GHZ^H$ may be written as a multiple sum. For simplicity, we identify the edge set with $[l]$ for some $l\in\mathbb{N}$. To each edge $e\in E$ we introduce a summation index $i_e$, and for each $j\in[k]$ we let $E_j$ be the set of edges incident with the vertex $j$. Then we can write
\begin{equation}
\GHZ^H_n=\sum_{i_1,\ldots,i_l=0}^{n-1}\ket{(i_e)_{e\in E_1}}_1\ket{(i_e)_{e\in E_2}}_2\cdots\ket{(i_e)_{e\in E_k}}_k
\end{equation}
Without loss of generality we will assume that there are no empty edges.

Following the idea in Strassen's proof, but more generally, we wish to consider several linear equalities involving the indices, and apply local $\epsilon$-dependent diagonal operations in such a way that the leading order contains precisely the terms satisfying the equalities. Such a system of equations can be written as
\begin{equation}
c_1i_1+c_2i_2+\cdots+c_li_l=g
\end{equation}
where $c_1,\ldots,c_l,g\in\mathbb{Z}^d$, $d\ge 1$. Subtracting $g$ from both sides and taking the inner product with itself results in
\begin{equation}
0=\langle g,g\rangle+\sum_{e\in E}\left(\langle c_e,c_e\rangle i_e^2-\langle c_e,g\rangle i_e\right)+\sum_{e,f\in E}\langle c_e,c_f\rangle i_ei_f
\end{equation}
We need to distribute the terms among the vertices in such a way that an index $i_e$ can only appear at vertices incident with $e$. This is always possible for the first term and the following sum, while the condition for the double sum is that $\langle c_e,c_f\rangle=0$ whenever there is no vertex incident to both $e$ and $f$. In other words, $c:E\to\mathbb{Z}^d$ is an orthogonal representation of the line graph $L(H)$ of $H$. Given such a $c$, we can choose local $\epsilon$-dependent operators $A_1(\epsilon),\ldots,A_k(\epsilon)$ in such a way that
\begin{multline}
(A_1(\epsilon)\otimes\cdots\otimes A_k(\epsilon))\ket{(i_e)_{e\in E_1}}_1\ket{(i_e)_{e\in E_2}}_2\cdots\ket{(i_e)_{e\in E_k}}_k  \\  =\epsilon^{\langle c_1i_1+\cdots+c_li_l-g,c_1i_1+\cdots+c_li_l-g\rangle}\ket{(i_e)_{e\in E_1}}_1\ket{(i_e)_{e\in E_2}}_2\cdots\ket{(i_e)_{e\in E_k}}_k,
\end{multline}
which shows that
\begin{equation}
\GHZ^H_n\degto\sum_{\substack{ 0\le i_1,\ldots,i_l\le n-1  \\  c_1i_1+\cdots+c_li_l=g}}\ket{(i_e)_{e\in E_1}}_1\ket{(i_e)_{e\in E_2}}_2\cdots\ket{(i_e)_{e\in E_k}}_k
\end{equation}

Next we want to ensure that the resulting state is a GHZ state. This happens precisely if the values of the indices at any one vertex determine the remaining ones uniquely. After fixing the indices at a vertex, $c_1i_1+\cdots+c_li_l=g$ becomes a system of linear equations in the remaining indices, therefore the condition is that the vectors corresponding to edges not incident with any one vertex are linearly independent. It follows that the dimension $d$ must be at least $|E|-\min_j|E_j|$, and therefore a sufficient condition is that any $d$ vectors are linearly independent. Orthogonal representations of $L(H)$ with this property are said to be in general position. \cite{LSS1}

We now show that after fixing the coefficients in this way, the right hand side $g$ can be chosen such that the number of solutions is large.
\begin{thm}
Let $H=([k],E,I)$ be a hypergraph and suppose that $c:E\to\mathbb{Z}^d$ is a general-position orthogonal representation of its line graph. Then
\begin{equation}
\omega(\GHZ^H_2,\GHZ_2)\le\frac{1}{|E|-d}.
\end{equation}
\end{thm}
\begin{proof}
Let $C$ be the maximum of $1$-norms of the vectors $\{c_e\}_{e\in E}$, and choose $G$ uniformly at random from the cube $[-Cn,Cn-1]^d\cap\mathbb{Z}^d$. Let $X_{(i_1,\ldots,i_l)}$ be the indicator random variable of the event that $(i_1,\ldots,i_l)$ is a solution of the (random) system of equations $c_1i_1+\cdots+c_li_l=G$. The number of solutions is
\begin{equation}
N=\sum_{i_1,\ldots,i_l=0}^{n-1}X_{(i_1,\ldots,i_l)}.
\end{equation}
Since
\begin{equation}
\mean X_{(i_1,\ldots,i_l)}=\prob(X_{(i_1,\ldots,i_l)}=1)=\prob(c_1i_1+\cdots+c_li_l=G)=(2Cn)^{-d},
\end{equation}
the expected number of solutions is
\begin{equation}
\mean N=n^l(2Cn)^{-d}=(2C)^{-d}n^{l-d}
\end{equation}
Therefore there is at least one vector $g$ such that $c_1i_1+\cdots+c_li_l=g$ has at least $M:=\lceil(2C)^{-d}n^{l-d}\rceil$ solutions. This implies that
\begin{equation}
\begin{split}
\omega(\GHZ^H_2,\GHZ_2)
 & \le \omega(\GHZ^H_2,\GHZ^H_n)\omega(\GHZ^H_n,\GHZ_M)\omega(\GHZ_M,\GHZ_2)  \\
 & \le \frac{\log n}{\log \lceil(2C)^{-d}n^{l-d}\rceil}\to\frac{1}{l-d}=\frac{1}{|E|-d}
\end{split}
\end{equation}
as $n\to\infty$.
\end{proof}

It follows that the best bound is obtained if $d$ is as small as possible. The following result completes the proof of Theorem \ref{thm:mainres} by showing that the optimal value is $d=|E|-\lambda(H)$.
\begin{thm}
Let $H=([k],E,I)$ be a hypergraph. Then its line graph has a general-position orthogonal representation $c:E\to\mathbb{Z}^{|E|-\lambda(H)}$.
\end{thm}
\begin{proof}
It is enough to see that there is an orthogonal representation $c:E\to\mathbb{Q}^{|E|-\lambda(H)}$, since after multiplying each vector by the least common denominator of its entries, it becomes one in $\mathbb{Z}^{|E|-\lambda(H)}$.

Since $H$ is $\lambda(H)$-edge-connected, its line graph is $\lambda(H)$-vertex-connected. A result by Lov\'asz, Saks and Schrijver \cite{LSS1,LSS2} states that any $(n-d)$-vertex-connected graph with $n$ vertices admits a general-position orthogonal representation in $\mathbb{R}^d$. Their proof relies in an essential way on using real numbers, and it does not seem to be possible to directly adapt the idea to our case. However, it is possible to deduce the existence of a general-position orthogonal representation in $\mathbb{Q}^d$ as follows.

Let $G=(V,E)$ be an $(n-d)$-vertex-connected graph with $n$ vertices, and let $V=\{v_1,v_2,\ldots,v_n\}$ be an ordering of the vertices. We construct a map $\mathcal{O}_G:(\mathbb{R}^d)^V\to(\mathbb{R}^d)^V$ recursively as follows. If $f:V\to\mathbb{R}^d$ is any function, then $(\mathcal{O}_Gf)(v_1)=f(v_1)$. If $i>1$, we let $A_i=\{v_j|j<i\text{ and }\{v_i,v_j\}\notin E\}$. Let $P_i$ denote the orthogonal projection onto the subspace spanned by $\{(\mathcal{O}_Gf)(v)|v\in A_i\}$. Then we set $(\mathcal{O}_Gf)(v_i)=(I-P_i)f(v_i)$.

It is clear from the definition that $\mathcal{O}_Gf$ is an orthogonal representation for any $f$, and that if $f$ is an orthogonal representation, then $\mathcal{O}_Gf=f$.

Let us find a more explicit form of the projections $P_i$. Let us write $A_i=\{v_{j_1},v_{j_2},\ldots,v_{j_r}\}$ with $j_1<j_2<\ldots<j_r$. To find $P_i$, we first orthogonalize the vectors $(\mathcal{O}_Gf)(v_{j_m})$, i.e. let $g_1=(\mathcal{O}_Gf)(v_{j_1})$ and
\begin{equation}
g_k=(\mathcal{O}_Gf)(v_{j_k})-\sum_{m=1}^{k-1}\frac{\langle g_m,(\mathcal{O}_Gf)(v_{j_k})\rangle}{\langle g_m,g_m\rangle}g_m
\end{equation}
for $k=2,\ldots,r$. Here and in the following sum we exclude the terms with $g_m=0$. Then the projection can be written as
\begin{equation}
P_ix=\sum_{m=1}^{r}\frac{\langle g_m,x\rangle}{\langle g_m,g_m\rangle}g_m.
\end{equation}
From this form we can see that if $f:V\to\mathbb{Q}^d$ then $P_i$ maps $\mathbb{Q}^d$ into itself, therefore $\mathcal{O}_Gf:V\to\mathbb{Q}^d$ as well.

If $f:V\to\mathbb{R}^d$ is a function such that none of the denominators in the above expression of $P_i$ is $0$, then $\mathcal{O}_Gf$ is clearly continuous at $f$. This holds in particular if $f$ is a general-position orthogonal representation, because in this case $\mathcal{O}_Gf=f$, and the families $(f(v))_{v\in A_i}$ are linearly independent since $|A_i|\le d$.

Now pick any general-position orthogonal representation $f:V\to\mathbb{R}^d$. By continuity at $f$ and using that being in general position is an open condition, there is an open neighbourhood $U$ of $f$ such that $f'\in U$ implies that $\mathcal{O}_Gf'$ is also a general-position orthogonal representation. Since $U\neq\emptyset$ is open, there is a function $c_0:V\to\mathbb{Q}^d$ in $U$, therefore $c:=\mathcal{O}_Gc_0$ has the desired properties.
\end{proof}

\section{Discussion}

Theorem \ref{thm:mainres} has a number of interesting special cases. Suppose first that the hypergraph $H$ is a graph consisting of a single path going through every vertex. Then clearly $\lambda(H)=1$, i.e. asymptotically one GHZ state per copy can be extracted via SLOCC. Of course, this can be easily proved without our result, since by teleportation the transformation can be performed on a single copy exactly via LOCC. On the other hand, if we add a single new edge joining the two endpoints, the graph becomes a cycle, which is 2-edge-connected, therefore asymptotically \emph{two} GHZ states per copy can be obtained. We do not know if the same can be accomplished via LOCC with asymptotically vanishing error.

For the next example, let $K_k^l$ denote the complete $l$-uniform hypergraph on $[k]$, i.e. the multiplicity of every $l$-element subset in $K_k^l$ is one and there are no other edges. Then it is not difficult to see that a minimum cut is obtained by removing every edge incident with a distinguished vertex, therefore
\begin{equation}
\omega(\GHZ^{K_k^l}_2,\GHZ_2)=\frac{1}{\binom{k-1}{l-1}}.
\end{equation}
This special case with $l=2$ as well as the result for cycle graphs have recently found applications in complexity theory. In particular, ref. \cite{Buhrman} shows that our results imply new protocols and bounds in nondeterministic multiparty quantum communication complexity.

Now let $H$ be an arbitrary hypergraph and suppose that instead of GHZ states, the parties wish to distill EPR pairs shared between a specified pair $AB$. Suppose that the minimum cut separating $A$ from $B$ consists of $t$ edges. By Menger's theorem for hypergraphs, there exist $t$ edge-disjoint paths $P_1,\ldots,P_t$ between $A$ and $B$ in $H$. Using Theorem \ref{thm:mainres} for the subhypergraph $P_i$, $\GHZ^{P_i}_2$ can be transformed to a GHZ state on the subset of vertices incident with at least one edge in $P_i$, which in turn can be converted to an EPR-pair between $A$ and $B$. This proves that $\omega(\GHZ^H_2,\EPR_{AB})=1/t$. This transformation also has an asymptotic LOCC counterpart with the same rate, under the name localizable entanglement \cite{SVW}.

Finally, let us mention that our result can be easily extended to products of GHZ-type states with possibly different number of levels. In this case, the minimum bipartite log-rank gives the asymptotic rate at which GHZ states can be obtained. Equivalently, $H$ may be replaced with a hypergraph with weighted edges, where the weight corresponding to an $r$-level GHZ state is $\log r$. In this case $\lambda(H)$ on the right hand side of eq. \eqref{eq:GHZrate} should be interpreted as the minimum cut weight.

\section{Acknowledgement}

We thank Jeroen Zuiddam for helpful discussions. MC acknowledges financial support from the European Research Council (ERC Grant Agreement no 337603), the Danish Council for Independent Research (Sapere Aude) and the Swiss National Science Foundation (project no PP00P2\_150734). Part of this work was done while the authors were with ETH Zurich.

\end{document}